\documentclass[letterpaper, 10 pt, conference]{ieeeconf}  

\IEEEoverridecommandlockouts                              
\usepackage{amsfonts}
\usepackage{amsmath}
\usepackage{mathtools}
\usepackage{subcaption}

\usepackage{color,soul}

\newtheorem{Theorem}{Theorem}
\newtheorem{Proposition}{Proposition}

\newtheorem{Lemma}{Lemma}

\title{\LARGE \bf
Maximizing Multivariate Information with Error-Correcting Codes
}

\author{Kyle Reing$^{*}$, Greg Ver Steeg, and Aram Galstyan
\thanks{*Correspondence: {\tt\small reing@usc.edu}}
\thanks{University of Southern California, Information Sciences Institute,
{\tt\small reing@usc.edu, gregv@isi.edu, galstyan@isi.edu}}}%

\begin{document}

\maketitle
\thispagestyle{empty}
\pagestyle{empty}

\begin{abstract}

Multivariate mutual information provides a conceptual framework for characterizing higher-order interactions in complex systems. Two well-known measures of multivariate information---total correlation and dual total correlation---admit a spectrum of measures with varying sensitivity to intermediate orders of dependence. Unfortunately, these intermediate measures have not received much attention due to their opaque representation of information. Here we draw on results from matroid theory to show that these measures are closely related to error-correcting codes. This connection allows us to derive the class of global maximizers for each measure, which coincide with maximum distance separable codes of order $k$. In addition to deepening the understanding of these measures and multivariate information more generally, we use these results to show that previously proposed bounds on information geometric quantities are tight at the extremes.


\end{abstract}

\section{INTRODUCTION}
\label{sec1}

Characterizing the interactions in a system by its pairwise correlations is a common modeling assumption, but a more complete picture would also consider many-to-one, and many-to-many interactions. In systems with emergent properties, such as those studied in ecology \cite{eco_syn}, systems biology \cite{sysbio_syn}, neuroscience \cite{neuro_syn}, and genetics \cite{gene_syn} (to name a few), these higher-order interactions play a vital role in macroscopic behavior. The issue of defining a proper measure of higher-order correlation, at least from an information theoretic perspective, is what the study of multivariate mutual information (MMI) seeks to do. Tools based on MMI and the related subject of information decomposition have already seen diverse application, offering insights into problems like structure discovery in biological/artificial neural networks  \cite{pid_nn1,pid_nn2,pid_nn4,pid_nn5,pid_nn6,pid_nn3,pid_conv}, and in determining how genes cooperate to produce a particular phenotypic effect \cite{pid_gene1,pid_gene2}. Two major challenges faced by any work attempting to utilize MMI are the selection of an appropriate measure, and its computational cost. The first concern is discussed in detail by Timme et al. \cite{exp_persp}, responding to the fact that no single measure is universally accepted by the community. Choosing from the multitude of measures, each with their own idiosyncratic definition of group information, often boils down to preference, application, or commitment to certain axiomatic principles. To add to this long list of considerations, the cost of analytically or approximately calculating the measure may be prohibitive in all but the smallest applications. The most expressive approaches (such as those based on the partial information decomposition framework \cite{pid}) are usually the most intractable, limiting their practical use to the simplest forms of multivariate dependence \cite{bivar2}. All this serves to reiterate that MMI is far from solved, and could benefit greatly from continued study, especially regarding the aforementioned computational problems. 

Our contribution to this ongoing dialogue is to identify the global maximizers for a class of MMI measures related to Watanabe's total correlation. Finding the extrema of a function is usually a matter of optimization, but for non-convex and/or high-dimensional functions, there may not be a guaranteed strategy for convergence to the optimal solution. Knowledge of a closed form generative procedure for the extrema of MMI is desirable for at least two reasons: 1) these properties inform us about the types of dependence preferentially treated by the measure, which is important when comparing different approaches; 2) additionally, the scalar value of MMI for any distribution becomes more meaningful when it can be ranked according to its relative distance to the min and max. 

The main theoretical results of this work connect the global maximizers for each of the $n-1$ measures ($n$ being the number of random variables) with the class of maximum distance separable (MDS) codes \cite{MW-S} (equiv. ideal secret sharing schemes \cite{Shamir}). Through this connection, it follows that a complete description of the set of global maximizers is still unknown, as this knowledge could be used to solve an open conjecture in coding theory and finite projective geometry \cite{projgeom} known as the MDS conjecture. The proof of this result draws on a few disciplines, including matroid theory and coding theory. Since we expect readers will come from varying backgrounds, we have tried to make the writing expository and the material self-contained. Section \ref{sec2} introduces the family of measures, summarizes the main contributions in non-technical terms, and runs through an example of why an uninformed search procedure is unlikely to discover the maximizers. The preliminaries of Section \ref{sec3} are split into two nearly disjoint subsections, with each presenting only the concepts necessary for both understanding and completion of the final proof. The proofs of the main claim appear in Section \ref{sec4}, building on the definitions, theorems, and intuitions of Sections \ref{sec2} and \ref{sec3}. Section \ref{sec5} discusses important related work, focusing on a family of information geometric measures shown to be upper-bounded by the measures studied here. Previous results on these bounds are strengthened in light of the connection to MDS codes. Finally, Section \ref{sec6} concludes with a discussion of where this work fits with regards to general trends in MMI research, highlighting the appearance of coding theory and cryptography in many recent papers. 

\section{COHESION MEASURES}
\label{sec2}
\subsection{Total Correlation and Dual Total Correlation}
\label{sec2:a}


Historically, the study of multivariate mutual information began with the introduction of two measures, one being Watanabe's \textit{\textbf{total correlation}} (TC) \cite{tc} (equiv. multiinformation \cite{multi}), defined as 
\begin{equation*}
TC(X) = \sum\limits_{i=1}^n H(X_i) - H(X).\\
\end{equation*}
Here the vector random variable $X = \{X_1, \ldots ,X_n \}$ will be taken over a finite support $||X||$, and entropy is defined in the standard way as $\sum\limits_{x\in ||X||}p(x) \log[ 1 / p(x)]$.  We adopt the convention that each marginal $X_i$ has the same support $q$ without loss of generality, implying $||X|| \leq q^n$. Additionally, we assume that our logarithms are taken with respect to base $q$.\\

A related, and equally important measure for our discussion is the \textit{\textbf{dual total correlation}} (DTC) (equiv. excess entropy \cite{excess}, binding information \cite{binding}), originating from Han's \cite{dtc} work on lattice theoretic duality of information measures. Dual total correlation can be defined as 
\begin{align*}
DTC(X) &= H(X) - \sum\limits_{i=1}^n H(X_i|X_{E/i})\\
&= \sum\limits_{i=1}^n H(X_{E/i}) - (n-1)H(X),
\end{align*}
where $X_{E/i}$ denotes the marginal random variable of cardinality $(n-1)$ that excludes $X_i$. Both total correlation and its dual share a number of desirable properties \cite{dtc}\cite{noneghan}, some of which are detailed below. 
\begin{itemize}
\item \textbf{\textit{Entropic:}} A measure is considered \textit{entropic} if it is defined as a function of subset entropies $H(X_A)$ for any $X_A \subseteq X$. Both total correlation and dual total correlation are \textit{linearly entropic} since they are given by linear functions $\langle \mathbf{C}, \mathbf{H} \rangle$, where $\mathbf{H}$ is a column vector of all subset entropies, $\mathbf{C}$ a row vector of $2^n$ constants, and $\langle \cdot,\cdot\rangle$ the inner product.
\item \textbf{\textit{Correlative:}} A measure is \textit{correlative} if it equals zero when the marginals $X_i$,\: $\forall i \in [n]$ are mutually independent.
\item \textbf{\textit{Symmetric:}} A measure is \textit{symmetric} if it remains unchanged under permutation of the marginals.
\item \textbf{\textit{Non-negative:}} The measure is always $\geq 0$.
\end{itemize}
\hfill

One property not shared between the two measures is that total correlation can be represented as the KL-Divergence $D_{KL}\big(p(x)||\prod\limits_{i=1}^n p(x_i)\big)$, whereas dual total correlation has no such divergence expression. Another difference is the set of maximizing distributions for each measure. Previous work has characterized these distributions \cite{binding}\cite{nihatmulti}, shown in Table \ref{tbl1:sub1}/\ref{tbl1:sub2} below for three variables with binary support. Maximizers of TC are exactly the distributions for which a bijection over support $q$ exists between every pair of marginal variables. Such a relationship is often referred to as \textit{informational redundancy}, making the distribution in Table \ref{tbl1:sub1} maximally redundant\footnote{We don't make the distinction here between pairwise and multivariate redundancy. Such a distinction is important in information decomposition, where shared and unique information \cite{uniq} are distinguished with respect to a target variable}. The counterpart to redundancy is \textit{informational synergy}, which refers to the presence of relationships that only exist from a collection $X_A$ ($1<|A|< n$) to subset $X_B \subseteq X_{E/A}$ of the complement set. The parity distribution in Table \ref{tbl1:sub2} only has dependence when all $n$ variables are considered, mapping every $(n-1)$ variable subset ($X_{E/i} \: ,\forall i \in [n])$ to the remaining variable $X_i$. Therefore, we can say that Table \ref{tbl1:sub2} is maximally synergistic of order $(n-1)$. 

\begin{table}[!htb]
    \caption{Maximizing Distributions for Three Variables}
	\begin{subtable}{.5\linewidth}
    	\centering
		\captionof{table}{Binary Maximizers TC}
		\begin{tabular}{ccc|c|}
			\boldmath$X_0$ & \boldmath$X_1$ & \boldmath$X_2$ &\boldmath$Pr$\\
			\hline
            0 & 0 & 0  &1/2 \\
            1 & 1 & 1  &1/2 \\
		\end{tabular}\label{tab:3}
        \label{tbl1:sub1}
	\end{subtable}%
	\begin{subtable}{.5\linewidth}
    	\centering
        \captionof{table}{Binary Maximizers DTC}
		\begin{tabular}{ccc|c|}
        	\boldmath$X_0$ & \boldmath$X_1$ & \boldmath$X_2$ &\boldmath$Pr$\\
			\hline
            0 & 0 & 0 &  1/4 \\
            0 & 1 & 1 &  1/4 \\
            1 & 0 & 1 &  1/4 \\
            1 & 1 & 0 &  1/4 \\
		\end{tabular}\label{tab:4}
        \label{tbl1:sub2}
	\end{subtable}
\end{table}

\subsection{Fujishige's Total Correlation(s) and Cohesion}
\label{sec2:b}
If TC is maximized by information contained in every pair of variables, and DTC by information present only in the $n$ variable joint state, what about information in intermediate collections of variables? Do any measures exist whose maximizers prefer dependence among subsets with cardinality between $2$ and $n$? A promising place to start in answering this question is a family of measures introduced by Fujishige \cite{polymatroid}. The main content of his paper focuses on connecting entropy and polymatroids (introduced in Section \ref{sec3:a}), but a small portion is dedicated to applying these insights to extend TC and DTC. His observation was that summing over all marginals of cardinality $1$ causes each member of the set to be covered $1$ time (or respectively, sums of cardinality $(n-1)$ cover each element $(n-1)$ times for DTC). These cover relations are reflected in the constant multiples of joint entropy appearing in both measures. From this, one can generalize to correlation measures which sum over marginals of cardinality $k$ and cover each element ${n-1 \choose k-1}$ times. These generalized measures, and their duals, are given by

\begin{equation*}
\mathcal{C}^{(k)}(X) = \sum\limits_{X_A \in \mathcal{E}_k} H(X_A) - {n-1 \choose k-1}H(X)
\end{equation*}
\begin{equation*}
\mathcal{C}^{(n-k)}(X) = \sum\limits_{X_B \in \mathcal{E}_{n-k}} H(X_B) - {n-1 \choose n-k-1}H(X).
\end{equation*}
Here, $k$ (and $(n-k)$ respectively) is called the \textit{interaction order}, which dictates the cardinality of subsets appearing in the first term
\begin{equation*}
\mathcal{E}_k = \{A : A \subseteq X, |A| = k \}.
\end{equation*}

For $n$ variables, $\mathcal{C}^{(1)}$ is clearly equivalent to the total correlation, and $\mathcal{C}^{(n-1)}$ to the dual total correlation. For those more comfortable with the conditional entropy definition of dual total correlation, each of the measures can be written as
\begin{equation*}
\begin{split}
\mathcal{C}^{(k)}(X) &= \sum\limits_{X_A \in \mathcal{E}_k} H(X_A) - {n-1 \choose k-1}H(X)\\
&= \sum\limits_{X_A \in \mathcal{E}_k} H(X_A) + \Big( {n-1 \choose k} - {n \choose k}\Big) H(X)\\
&={n-1 \choose k} H(X) - \sum\limits_{X_B \in \mathcal{E}_{n-k}} H(X_B|X_A)\\
\end{split}
\end{equation*}

The desirable properties satisfied by $\mathcal{C}^{(1)}$ and $\mathcal{C}^{(n-1)}$, including being \textit{linearly entropic}, continue to be satisfied for any choice of parameters $n$ and $k$. Because of this, the measures can be related according to the following linear inequalities \cite{polymatroid}, which we call \textit{polymatroid bounds}. 
\begin{equation}
\begin{split}
(n-k){\cdot}\mathcal{C}^{(k)}(X) \:\geq\: k {\cdot} \mathcal{C}^{(k+1)}(X)\\
(n-k){\cdot}\mathcal{C}^{(n-k)}(X) \:\geq \: k {\cdot}\mathcal{C}^{(n-k-1)}(X)
\end{split}
\label{eq1}
\end{equation}

\begin{figure*}
  \includegraphics[width=\textwidth]{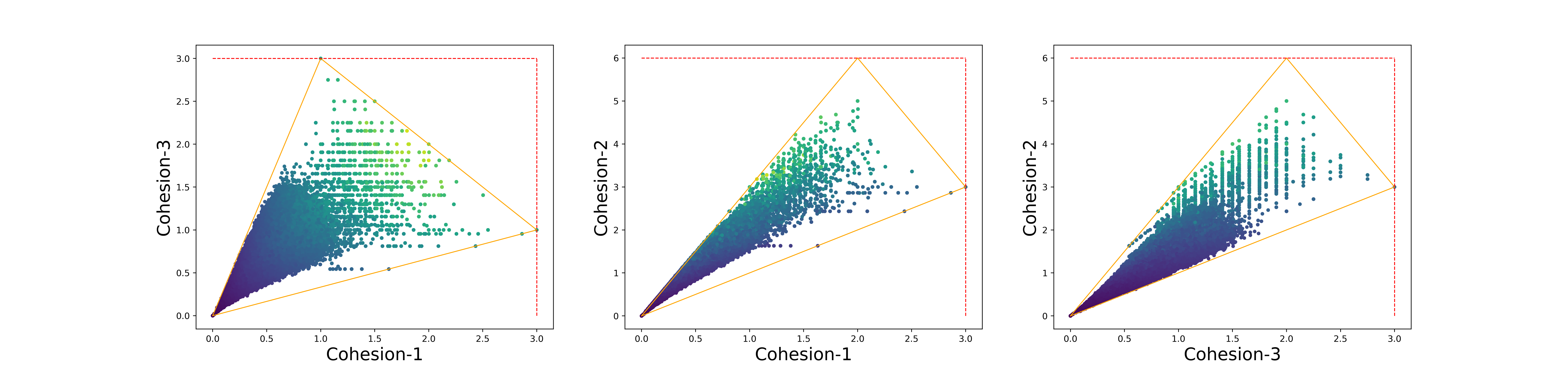}
  \caption{Plots of each Cohesion measure for four binary variables, alongside the polymatriod bounds (orange lines). The color gradient represents the value of the third measure, with blue for low values, and yellow for high. The dashed red lines correspond to the constant upper bounds, introduced in Section \ref{sec4}}
  \label{fig1}
\end{figure*}

Fujishige referred to each measure as a (dual) total correlation, but we believe such overlapping nomenclature might be confusing. To maintain some link to total correlation without trying to re-brand existing terminology, we refer to the family of measures as \textit{\textbf{Cohesion measures}}, after one of Watanabe's (and subsequently, Han's) original names for the total correlation \cite{dtc,cohesion}. For interaction order $k$, the associated measure is referred to as \textit{\textbf{Cohesion-k}}, or adopting the shorthand of \cite{nihat_geom}, $\mathcal{C}^{(k)}$.

\subsection{Summary of main results}
\label{sec2:c}
The maximizers for \textit{Cohesion-k} have only been studied for the special cases of $k = 1$ and $(n-1)$. This leaves the question open regarding each Cohesion measures' sensitivity to certain orders of dependence. We will show through the main result of the paper (Theorem 3) that each intermediate measure can be upper-bounded by a constant, and that this bound is always achieved for a particular class of distributions with marginal support $q$. The exact value of $q$, if known in general, would provide an answer to the MDS conjecture (discussed further in Section \ref{sec3}). Intuitively, the MDS conjecture is a claim about how large $q$ must be ($q \geq n-1$) in order for certain types of linear dependence structures to occur. Despite the uncertainty surrounding $q$, it is still possible to prove that such an integer must exist for arbitrary values of $n$ and $1\leq k \leq n-1$. The proof relies on a subfield of combinatorics, called matroid theory, to exploit the dependence structure of distributions that achieve the bound with equality. Matroids are objects that generalize the concepts of linear independence and provide a common language for talking about dependence, both in codes and probability distributions.\\

While matroids are required to make the most general claims, for special cases, such as when the number of variables is a prime power and $q=n$, straightforward methods for constructing MDS codes are well established. Because these special cases provide concrete examples of distributions that achieve the bound, the details of their construction is worth introducing. Thus, the preliminaries of Section \ref{sec3} attempts to balance high level abstractions required for the proof (Section \ref{sec3:a}) with motivated examples from coding theory (Section \ref{sec3:b}). Before jumping straight into these preliminaries, we wish to solidify intuition for the problem, and further demonstrate why the theoretical approach is necessary. To do this, we look at the simplest Cohesion measure whose maximizers have not been characterized, which is \textit{Cohesion-2} over four random variables.

\subsection{Cohesion-2 for 4 Variables}
\label{sec2:d}

A naive but immediately accessible way of exploring the maximizing distributions, given the material presented so far, is empirical evaluation. This might take the form of a simplex search using gradient descent, or even a brute force evaluation of every point on a discretized (alt. uniformly sampled) simplex. Due to the low dimensionality of this special case, we have the luxury of performing the latter without much computational concern. Assuming for now that each marginal variable is binary, and given a sample from the $16$-simplex, we can calculate all three Cohesion measures for four variables and use the results as coordinates to a three-dimensional vector space.\\

Through the polymatroid upper and lower bounds of Equation \ref{eq1}, we can define a convex polytope over feasible solutions. Normally the region is unbounded, but constraints on the support of each variable lead to three additional inequalities
\begin{equation*}
\mathcal{C}^{(1)} + \mathcal{C}^{(3)} \leq 4, \quad \mathcal{C}^{(2)} + 3 {\cdot} \mathcal{C}^{(1)} \leq 12, \quad \mathcal{C}^{(2)} + 3 {\cdot}\mathcal{C}^{(3)} \leq 12.
\end{equation*}

The proofs for these inequalities appear in Appendix A. Figure \ref{fig1} shows the two-dimensional projections of this space, alongside the bounds. Given these results, we are interested in two questions: what are the empirically maximal distributions, and do they achieve the upper bound? As a sanity check, we see \textit{Cohesion-1} (TC) and \textit{Cohesion-3} (DTC) achieve their maximum of $3$ bits for the four variable versions of distributions in Table 1, which meet the bounds with equality. \textit{Cohesion-2}, on the other hand, peaks for the distribution in Table \ref{tbl2}, which has been called a redundant-synergy distribution \cite{griffith2}, owing to the combination of third-order (first three variables) and second-order (first and last variable) information. This answers the first question, but when comparing the value of this distribution ($5$ bits) with the upper bound ($6$ bits), we observe a large gap in the feasible region. The appearance of this gap unfortunately means that there is no verification this distribution is actually maximal. Perhaps the sampling over the simplex was too course, or the bound is too loose and additional (possibly nonlinear) inequalities are required.
\begin{table}[!htb]
    	\centering
        \caption{Binary Maximizers for Cohesion-2}
		\begin{tabular}{cccc|c|}
        	\boldmath$X_0$ & \boldmath$X_1$ & \boldmath$X_2$ & \boldmath$X_3$&\boldmath$Pr$\\
			\hline
            0 & 0 & 0 & 0 & 1/4\\
            0 & 1 & 1 & 0 & 1/4\\
            1 & 0 & 1 & 1 & 1/4\\
            1 & 1 & 0 & 1 & 1/4\\
		\end{tabular}
        \label{tbl2}
\end{table}
Suppose we are fortunate in not having to check these alternatives due to our access to an oracle $\mathcal{O}(\mathcal{C},n)$. This oracle takes as input our Cohesion measure $\mathcal{C}$, some integer number of variables $n$, and returns a distribution of cardinality $n$ in the set of global maximizers for $\mathcal{C}$. On entering \textit{Cohesion-2} and $n=4$, the following distribution is revealed (Table \ref{tbl3}).\\

\begin{table}[!htb]
    	\centering
        \caption{Reed-Solomon Maximizers for Cohesion-2}
		\begin{tabular}{cccc|c|cccc|c|}
        	\boldmath$X_0$ & \boldmath$X_1$ & \boldmath$X_2$ & \boldmath$X_3$&\boldmath$Pr$&\boldmath$X_0$ & \boldmath$X_1$ & \boldmath$X_2$ & \boldmath$X_3$&\boldmath$Pr$\\
			\hline
            0 & 0 & 0 & 0 & 1/16 &2 & 2 & 2 & 2 & 1/16\\
            0 & 1 & 2 & 3 & 1/16 &2 & 3 & 0 & 1 & 1/16\\
            0 & 2 & 3 & 1 & 1/16 &2 & 0 & 1 & 3 & 1/16\\
            0 & 3 & 1 & 2 & 1/16 &2 & 1 & 3 & 0 & 1/16\\
            1 & 1 & 1 & 1 & 1/16 &3 & 3 & 3 & 3 & 1/16\\
            1 & 0 & 3 & 2 & 1/16 &3 & 2 & 1 & 0 & 1/16\\
            1 & 3 & 2 & 0 & 1/16 &3 & 1 & 0 & 2 & 1/16\\
            1 & 2 & 0 & 3 & 1/16 &3 & 0 & 2 & 1 & 1/16\\
		\end{tabular}
        \label{tbl3}
\end{table}
Calculating \textit{Cohesion-2} for this distribution yields a value of $6$ quaternary digits ($12$ bits), meeting the upper bound with equality. Our oracle has verified that the polymatroid bounds are indeed enough to fully characterize Cohesion, at least in this special case. However, we've only replaced one mystery with another: how does the oracle work, and does it generalize to arbitrary Cohesion measures? It turns out that each row in Table \ref{tbl3} is related to an important concept in algebraic coding theory called a Reed-Solomon code. Reed-Solomon codes have the property of being maximum distance separable (MDS), and provide a way of constructing these codes when the support is large (\:$q \geq n$\:) and $n$ is a prime power. Notice that our restriction to binary support in Figure \ref{fig1} played an important role in the existence of a gap. The dangers of empirical evaluation should now be apparent if they weren't before. Even if the criteria on the number of variables and size of the support are met, a search over the $n^n$ simplex would be required to find a Reed-Solomon distribution. As we will later show, if the MDS conjecture holds true, for any value of $n$ and $k$, $q$ is expected to be at least $n-1$, meaning superexponential search complexity is likely unavoidable. Thankfully, since we know which vertex of the Cohesion polytope a maximizer would appear if one existed, we can abandon search in favor of investigating the structural properties at this point. The next section is dedicated to preliminaries on matroid and coding theory, including illustrative examples whenever possible. 

\section{PRELIMINARIES}
\label{sec3}
\subsection{Matroids and Polymatroids}
\label{sec3:a}
A \textbf{\textit{matroid}} $M$ is given by a tuple $(E,\mathcal{I})$, where $E$ is some finite set, and $\mathcal{I}$ is a collection of subsets of $E$ satisfying the following three conditions:

\begin{description}
\item [\textit{(M1)}:]$\emptyset \in \mathcal{I}$
\item [\textit{(M2):}]If $I \in \mathcal{I}$ and $I' \subseteq I$, then $I' \in \mathcal{I}$
\item [\textit{(M3):}] If $I_1$ and $I_2$ are in $\mathcal{I}$ and $|I_1| < |I_2|$, then there is an element $e$ of $I_2 - I_1$ such that $I_1 \cup e \in \mathcal{I}$
\end{description}

$E$ and $\mathcal{I}$ are referred to as the \textbf{\textit{ground set}} and the \textbf{\textit{independent set}} of $M$, respectively. As an example, consider an $m\times n$ matrix $A$. Define the ground set $E(A)$ to be the set of all column labels $\{i \in [n]\}$ indexing $A$. The independent set $\mathcal{I}(A)$ consists of all collections of labels that index linearly independent columns. For the matrix $A$ given below, the independent set would be $\mathcal{I}(A) = \big\{\emptyset,\{1\},\{2\},\{3\},\{1,2\},\{1,3\},\{2,3\}\big\}$.
\[
A\: =
\begin{array}{c}
\begin{matrix}
\mathbf{1} & \mathbf{2} & \mathbf{3}
\end{matrix} \\
\left[\ \begin{matrix}
1 & 0 & 1 \\
0 & 1 & 1 
\end{matrix}\ \right]
\end{array}
\]

Any matroid that can be used to describe the dependence structure of a matrix defined over some finite field $\mathbb{F}_q$ (as with $A$ above for $q=2$) is called a \textit{\textbf{vector matroid}}. The last axiom \textit{(M3)} is called the \textbf{\textit{exchange axiom}}, and is the least intuitive of the three. To get a better sense for it, we include a proof that it is satisfied for the independent set of a vector matroid.

\begin{proof}
Let $I_1$ and $I_2$ be two members of the independent set $\mathcal{I}$ for a vector matroid $M(A)$, with $|I_1|<|I_2|$. This means both $I_1$ and $I_2$ index collections of linearly independent columns. The subspace $\mathcal{W}$ spanned by $I_1 \cup I_2$ has dimension at least $|I_2|$ (which occurs if $I_1 \subset I_2$). Assume for the sake of contradiction that $I_1 \cup e$ is linearly dependent for all $e \in I_2 - I_1$. Subtraction here is set subtraction, making $I_2 - I_1$ the set of all elements contained in $I_2$, but not in $I_1$. By this assumption, $\mathcal{W}$ is contained in the span of $I_1$, implying $|I_2|\leq dim \mathcal{W} \leq |I_1| < |I_2|$; contradiction. 
\end{proof}
\hfill\\

For two matroids $M_1$ and $M_2$, denote the corresponding ground sets by $E(M_1)$, $E(M_2)$, and independent sets by $\mathcal{I}(M_1)$, $\mathcal{I}(M_2)$. These matroids are said to be \textbf{\textit{isomorphic}} $(M_1 \cong M_2)$ if there exists a bijection $\psi$ satisfying
\begin{equation*}
\psi : E(M_1) \rightarrow E(M_2),\quad\forall X\subseteq E(M_1)
\end{equation*}
\begin{equation*}
 \quad \psi(X) \in \mathcal{I}(M_2) \iff X \in \mathcal{I}(M_1).
\end{equation*}
If a matroid $M_1$ is isomorphic to a vector matroid $M_2(A)$, where $A$ is a matrix over the field $\mathbb{F}_q$, then $M_1$ is said to be $\mathbf{F_q}$\textit{\textbf{-representable}}. It's important to note that not every matroid is $\mathbb{F}_q$-representable for an arbitrary field. To illustrate this, let $\mathbb{F}_2$ be the Galois field of two elements, where element $x \in \{0,1\}$, and addition/multiplication are given by mod-2 arithmetic. Construct a matroid $M$ over $E = \{1,2,3,4 \}$ such that all subsets of cardinality two or less are contained in the independent set; namely, 
\begin{equation*}
\mathcal{I} = \big\{S\subseteq E : 0 \leq |S| \leq 2 \big\}.
\end{equation*}
Any matroid defined over a ground set of cardinality $n$, whose independent set contains all subsets of cardinality $k$ or less is called a \textit{\textbf{uniform matroid}} $U_{k,n}$. Here, $M$ corresponds to the uniform matroid $U_{2,4}$.\\
\begin{Proposition}
[Oxley] The uniform matroid $U_{2,4}$ is not $\mathbb{F}_2$-representable
\end{Proposition}
\begin{proof}
Assume that $U_{2,4}$ is $\mathbb{F}_2$-representable, implying that $U_{2,4} \cong M(A)$ for some matrix $A$ defined over $\mathbb{F}_2$. Since the largest element in $\mathcal{I}(U_{2,4})$ has cardinality two, the column space of $A$ must have dimension two. A two-dimensional vector space over $\mathbb{F}_2$ has exactly four members, but only three of them are non-zero.  This means that $A$ cannot have four distinct non-zero columns, so a set of two columns in $A$ must be linearly dependent. However, this contradicts the original claim, since every pair of columns must be linearly independent in order for $U_{2,4} \cong M(A)$.
\end{proof}
\hfill

A \textit{\textbf{polymatroid}} $\mathcal{P}$ is given by a tuple $(E,f)$ of ground set $E$ and a \textbf{\textit{submodular set function}} $f$. A submodular set function is a mapping from the power set of $E$ to the non-negative real numbers ($f : 2^{E} \rightarrow \mathbb{R}^+$) satisfying the following three conditions:

\begin{description}
\item [\textit{(non-negativity):}] \qquad\qquad\:\,$f(\emptyset) = 0$, $ f(S) \geq 0 \quad \forall S \subseteq E$
\item [\textit{(monotonicity):}]\qquad\qquad\,\,$f(S) \leq f(T) \iff S \subseteq T \subseteq E$
\item [\textit{(submodularity):}]\qquad\qquad\,$f(S \cup  T) + f(S \cap T) \leq f(S) + f(T)$ 
\end{description}
Fujishige \cite{polymatroid} was the first to show that Shannon entropy acts as a submodular set function $\mathcal{H}$ for the polymatroid $(E,\mathcal{H})$ over ground set $E$ of random variables $\{X_1,X_2,...,X_n\}$. Polymatroids are generalizations of matroids in that, if we restrict the codomain of the set function to the non-negative integers $(f: 2^E \rightarrow \mathbb{Z^+})$ and require the additional upper bound condition $f(S) \leq |S| \:$ $\forall S \subseteq E$, then $f$ is called a \textit{\textbf{rank function}}, and can be used to build a matroid $M$.\\
\begin{Lemma}
[Oxley]Let $E$ be a set and $f$ a function on $2^E$ satisfying the above rank function conditions. If $X$ and $Y$ are subsets of $E$ such that, for all $y \in (Y-X)$,\: $f(X\cup y) = f(X)$, then $f(X\cup Y) = f(X)$
\end{Lemma}
\begin{proof}
See \cite{oxley}, the argument follows by induction on the number of elements in $(Y-X)$.
\end{proof}

\hfill

\begin{Theorem}
[Oxley] Let $E$ be a set and $f$ a function on $2^E$ satisfying the above rank function conditions. Let $\mathcal{I}$ be the collection of subsets $S$ of $E$ for which $f(S) = |S|$. Then $(E,\mathcal{I})$ is a matroid having rank function $f$.
\end{Theorem}
\begin{proof}
We will show that the independent set $\mathcal{I}$, constructed in the manner above satisfies the matroid axioms \textit{(M1)-(M3)}.
\begin{itemize}
\item By the non-negativity and upper bound conditions on $f$,
$0\leq f(\emptyset)\leq |\emptyset| = 0$, so $f(\emptyset) = |\emptyset|$ and $\emptyset \in \mathcal{I}$, satisfying \textit{(M1)}.
\item If $I \in \mathcal{I}$ then $f(I) = |I|$. For $I' \subseteq I$, and by submodularity of $f$, $f(I) + f(\emptyset) \leq f(I') + f(I-I')$. Since each term on the right hand side is upper-bounded, we can simplify to $|I| \leq |I'| + |I - I'| = |I|$. Equality must hold throughout, implying $f(I') = |I'|$ and $I' \in \mathcal{I}$, satisfying \textit{(M2)}.
\item Assume for the sake of contradiction $I_1$ and $I_2$ are in $\mathcal{I}$ with $|I_1|<|I_2|$,  but for all $e \in I_2 - I_1$, $f(I_1 \cup e) \neq |I_1 \cup e|$. We have $|I_1| + 1 > f(I_1 \cup e) \geq f(I_1) = |I_1|$, meaning $f(I_1 \cup e) = |I_1|$. Now, by applying Lemma 1 with $X = I_1$ and $Y = I_2$, it follows that $f(I_1) = f(I_1 \cup I_2)$. But $f(I_1 \cup I_2) \geq f(I_2) = |I_2|$, so $|I_1| \geq |I_2|$; contradiction, meaning \textit{(M3)} is satisfied. 
\end{itemize}
 
\end{proof}

With these basic definitions and theorems established, we move on to our discussion of coding theory.

\subsection{Algebraic Coding Theory and MDS Codes}
\label{sec3:b}

The fundamental object of coding theory is the message, which in the discrete case considered here is simply a string of length $k$ represented over some alphabet $\Sigma$ of cardinality $q$. Different subfields of coding theory are often interested in transformations of messages into codewords (strings of length $n \geq k$), with the utility of these codewords ranging from error-correction to cryptographic secret sharing \cite{SSSurvey}. For each mapping from message to codeword, the message length $k$, codeword length $n$, alphabet size $q$, and minimum distance $d$, are important parameters. The minimum distance is given by $d := \min\limits_{c_i,c_j \in \mathcal{C}} \Delta(c_i,c_j)$, where $\Delta(\cdot)$ is the Hamming Distance, and $c_i$ $c_j$ are any two codewords in $\mathcal{C} \subseteq q^n$. This minimum distance is used as a measure of error-correction and detection capabilities of a code. In particular, if the minimum distance between any two codewords is $k+1$, the coding scheme can be used to detect $k$ errors, while a minimum distance of $2k + 1$ allows for the correction of $k$ errors \cite{MW-S}.\\

For any coding scheme with minimum distance $d$, we can upper bound the number of codewords $\mathcal{M}$ in the following way. For the codewords $c_1,...,c_{\mathcal{M}}$ of an $(n,k,d)_q$ code $\mathcal{C}$, let $\bar{c}_i$ be the prefix of $c_i \in \mathcal{C}$ of length $n-d+1$. Each $\bar{c}_i$ must be distinct $(\bar{c}_i \neq \bar{c}_j)$, otherwise $\Delta(c_i,c_j) \leq d-1$, violating the claim that $\mathcal{C}$ has minimum distance $d$. Thus, the number of these prefixes bounds $\mathcal{M} \leq q^{n-d+1}$, which is known as the \textit{Singleton bound}. For linear codes, which are any $\mathcal{C} \subseteq q^n$ such that each codeword $c_i \in \mathcal{C}$ is generated by a linear combination of codewords $c_j \in \mathcal{C}$ ($c_i \neq c_j$), $\mathcal{M}$ is equal to $q^k$. Plugging this into the Singleton bound, we get $d \leq n - k + 1$. Codes that meet this bound with equality are termed maximum distance separable (MDS), and are an important class of codes with optimal error-correction/detection capabilities (for large alphabets). If we treat each $c_i$ as a $1 \times n$ row vector, MDS codes are represented by a basis $c_1,...,c_k$ of $k$ vectors, whose span produces all $q^k$ codewords. This $k\times n$ collection of basis vectors is called the \textit{\textbf{generator matrix}} $\mathcal{A}_{k,n}$ of $\mathcal{C}$. An interesting property of $\mathcal{A}_{k,n}$ follows from the constraints imposed by the minimum distance $d = n-k +1$; each set of $k$ (and by extension, $\leq k$) columns are linearly independent. This fact falls directly from the prefix argument introduced above.\\

Recalling examples from the matroid preliminaries, linearly independent columns can be used to define the independent set of a vector matroid. In this case the matroid produced by $\mathcal{A}_{k,n}$ is isomorphic to the uniform matroid $U_{k,n}$. This will be the crucial fact allowing us to connect codes and probability distributions, but we must first discuss the alphabet size $q$ (or equivalently, the field $\mathbb{F}_q$ of $\mathcal{A}_{k,n}$) necessary for a code to be MDS. The MDS conjecture postulates that MDS codes only exist when $q \geq n-1$, except when $k=3$  or $k = q -1$ and $q=2^m$  $m \in \mathbb{Z}^{>0}$, in which case $q \geq n-2$ \cite{projgeom}\cite{MDS_conj}. While such conditions seem arbitrary at first glance, they stem from a meaningful equivalence with an object in projective geometry called a \textit{k-arc}. Most progress on the MDS conjecture, such as the recent positive result for prime fields \cite{ball_mds}, proceeds by arguments and new results on these geometric objects.\\

As previously mentioned, when $q$ is large, strategies for constructing certain MDS codes are known, with the most well established corresponding to Reed-Solomon (RS) codes \cite{RScodes}. In explaining what Reed-Solomon codes are, and how they are generated, we will only consider the classical case when $q = n$. Note that this limits the possible values of $n$ significantly, since finite fields $\mathbb{F}_q$ only exist when $q= p^m$ is a prime power. In these cases, the elements of $\mathbb{F}_{p^m}$ are the $p^m$ roots of the polynomial $x^{(p^m)} - x$ \cite{lidl}. In the simplest case when $q$ is a prime number ($m=1$), the elements of $\mathbb{F}_p$ are just $\mathbb{Z}_p$, the integers mod $p$ . Given values for $k$, and $q=n$, start by defining a polynomial
\begin{equation*}
f(z) = \sum\limits_{j =0}^{k-1}f_j z^j \quad \forall_j f_j \in \mathbb{F}_q
\end{equation*}
of degree less than $k$ for some indeterminate $z$. Since the coefficients range over all $k$-tuples in $(\mathbb{F}_q)^k$, $f(z)$ is one of $q^k$ polynomials of this form. Next, define a $q$-tuple $\mathcal{B} = (\beta_1,\beta_2,...,\beta_q)$ such that $\beta_1=0, \beta_2 = 1, \beta_3 = \alpha, ..., \beta_q = \alpha^{q-2}$, for a primitive element $\alpha$ over $\mathbb{F}_q$. A primitive element is any element of the field that forms a multiplicative cyclic group, meaning raising $\alpha$ to some power generates all other elements (except 0). This implies $\mathcal{B}$ is just an ordered list of distinct elements in $\mathbb{F}_q$ (ex: a permutation of $[q]-1$ for $q$ a prime number). The idea behind introducing $\mathcal{B}$ is that we wish to construct another polynomial of degree less than $n$ by evaluating the polynomials $f(z)$ at $n$ points. This procedure is referred to as a \textit{\textbf{valuation map}} from $(\mathbb{F}_q)^k$ to $(\mathbb{F}_q)^n$. By evaluating $f(z)$ at each element in $\mathcal{B}$, $\big(f(\beta_1),...,f(\beta_q)\big)$, with each $f(\beta_i) = \sum\limits_{j=0}^{k-1} f_j \beta_i^j$, we can produce the $k$ basis vectors for our generator matrix $\mathcal{A}_{k,n}$. These vectors form a Vandermonde matrix, where each row is $\mathcal{B}$ raised to some power.\\
\[
\mathcal{A}_{k,n}\: =
\begin{bmatrix}
\:1 & 1 & 1 & \dots & 1 \\
\:0 & 1 & \alpha & \dots & \alpha^{-1}\\
\:0 & 1 & \alpha^2 & \dots & \alpha^{-2}\\
\:\vdots & \vdots & \vdots & \vdots & \vdots\\
\:0 & 1 & \alpha^{k-1} & \dots & \alpha^{-(k-1)}
\end{bmatrix}\
\]

Here the negative exponents are the same as a reverse indexing of $\mathcal{B} - \{0\}$ (ex: $\alpha^{-1} = \alpha^{q-2}$ ). To build intuition as to why/how this procedure works, let's look at the Reed-Solomon code underlying the distribution in Table \ref{tbl3}. In this case, $n=4$ with field $\mathbb{F}_4 = \mathbb{F}_{2^2}$. The elements of $\mathbb{F}_4$ (the roots of $x^4 - x$) are $\{0,1,z,z+1\}$ for indeterminate $z$. Addition and multiplication are given by the following tables, where $z^2 = z+1$.

\begin{table}[!htb]
	\begin{subtable}{.5\linewidth}
    	\centering
		\captionof{table}{Addition over $\mathbb{F}_4$}
		\begin{tabular}{c|cccc}
			\boldmath$\oplus$ &\boldmath$0$ &\boldmath$1$ &\boldmath$z$ &\boldmath$z^2$\\
			\hline
            \boldmath$0$ &$0$ & $1$  &$z$ & $z^2$ \\
            \boldmath$1$ & $1$ & $0$  &$z^2$ & $z$ \\
            \boldmath$z$ &$z$ & $z^2$  &$0$ & $1$ \\
            \boldmath$z^2$ &$z^2$ &$z$ & $1$  &$0$ \\
		\end{tabular}\label{tab:3}
        \label{tbl4:sub1}
	\end{subtable}%
	\begin{subtable}{.5\linewidth}
    	\centering
        \caption{Multiplication over $\mathbb{F}_4$}
		\begin{tabular}{c|cccc}
        	\boldmath$\otimes$ &\boldmath$0$ &\boldmath$1$ &\boldmath$z$ &\boldmath$z^2$\\
            \hline
            \boldmath$0$ &$0$ & $0$ & $0$  &$0$ \\
            \boldmath$1$ &$0$ & $1$ & $z$  &$z^2$ \\
            \boldmath$z$ &$0$ & $z$ & $z^2$  &$1$ \\
            \boldmath$z^2$ &$0$ & $z^2$ & $1$  &$z$ \\
		\end{tabular}\label{tab:4}
        \label{tbl4:sub2}
	\end{subtable}
\end{table}

Since $4$ is a prime power, multiplication over the field obeys mod-$g(z)$ arithmetic, for a prime polynomial $g(z)$ of degree $m=2$ (in particular, $g(z) = z^2 + z + 1$). We have chosen not to include a detailed background on non-prime finite fields (which can be found here \cite{lidl}), but only knowledge of the addition/multiplication tables is required to move forward. For $k=2$, the polynomials $f(z)$ of the valuation map take the form $f_0 + f_1 z, \: \{f_0,f_1\} \in \mathbb{F}_4$, and $\mathcal{B} = \{0,1,z,z+1\}$. Evaluating $f(z)$ at each $\beta \in \mathcal{B}$ results in the vector $\big(f_0, f_0 + f_1, f_0 + f_1z, f_0 + f_1 (z+1) \big)$. Choosing the standard basis, i.e $f(z) = 1$ and $f(z) = z$, the vector simplifies to $(1,1,1,1)$ and $(0,1,z,z+1)$ respectively. The span of these produces the following $16$ row vectors.

\begin{table}[!htb]
    	\centering
        \caption{Reed-Solomon codes over $\mathbb{F}_4$}
		\begin{tabular}{|cccc|c|cccc|}
        	
            $0$ & $0$ & $0$ & $0$ &  &$z$ & $z$ & $z$ & $z$ \\
            $0$& $1$ & $z$ & $z+1$ &  &$z$ & $z+1$ & $0$ & $1$ \\
            $0$ & $z$ & $z+1$ & $1$ &  &$z$ & $0$ & $1$ & $z+1$\\
            $0$ & $z+1$ & $1$ & $z$ &  &$z$ & $1$ & $z+1$ & $0$ \\
            $1$ & $1$ & $1$ & $1$ &  &$z+1$ & $z+1$ & $z+1$ & $z+1$ \\
            $1$ & $0$ & $z+1$ & $z$ &  &$z+1$ & $z$ & $1$ & $0$ \\
            $1$ & $z+1$ & $z$ & $0$ &  &$z+1$ & $1$ & $0$ & $z$ \\
            $1$ & $z$ & $0$ & $z+1$ &  &$z+1$ & $0$ & $z$ & $1$ \\
		\end{tabular}
        \label{tbl5}
\end{table}

Since all fields with $p^m$ elements are isomorphic to $\mathbb{F}_{p^m}$ up to relabeling of the elements, we can map $z$ to $2$ and $z+1$ to $3$. Doing so produces the rows of Table \ref{tbl3}.
\section{MAXIMIZING DISTRIBUTIONS}
\label{sec4}
To recapitulate the last two sections, we introduced matroids and demonstrated their connection to MDS codes (through the generator matrix $\mathcal{A}_{k,n}$) and to entropy (through polymatroids). What's left to show is that the maximizers of \textit{Cohesion-k} have a matroidal structure isomorphic to that of MDS codes, and that $\mathbb{F}_q$-representability of an MDS code implies achievability of the polymatroid bound for distributions with support $q$. To prove the first claim, we introduce a constant upper bound on each Cohesion measure that meets the polymatroid bound at a single point. Given the standard expression for \textit{Cohesion-k}, namely
\begin{equation*}
\mathcal{C}^{(k)}(X) = \sum\limits_{X_A \in \mathcal{E}_k} H(X_A) - {n-1 \choose k-1}H(X),
\end{equation*}
we can upper bound each term in the sum, and lower bound the joint entropy. For any $S \subseteq X$, the entropy $H(S)$ is maximal when $S$ is uniform over its support. Since each $X_i$ has support $q$, $S = \{X_1,...,X_m\},\: ||S|| \leq q^m$, it follows that $H(S) \leq \log_{q}(q^m) = m$, the cardinality of $S$. Thus, each term in the sum of \textit{Cohesion-k} is at most $k$. The joint entropy can be lower bounded by noticing $X_A \subseteq X$ and $H(X_A) \leq H(X)$ by monotonicity, implying
\begin{equation*}
H(X) \geq \max_{X_A \in \mathcal{E}_k}H(X_A) = k.
\end{equation*}
Substituting these results into the equation, we get the following constant bound
\begin{equation*}
\mathcal{C}^{(k)}(X) \leq k {n \choose k} - k {n-1 \choose k-1} = k {n-1 \choose k}.
\end{equation*}

\begin{Lemma}
The entropy function at the constant bound maps the power set of random variables in $X$ to the non-negative integers $(\mathcal{H}: 2^X \rightarrow \mathbb{Z}^+)$
\end{Lemma}
\begin{proof}
Split up the elements of the power set into different cases according to the value of k. Consider some subset $S \subseteq X$ such that the cardinality of $S$ is
\begin{itemize}
\item $\mathbf{0} :  H(\emptyset) = 0,$ which is integer valued. 
\item $\mathbf{j, \: 1 \leq j \leq k } :$ Since the upper bound requires all subsets $X_A$ of cardinality $k$ to have $H(X_A) = k$, every collection of $k$ random variables is independent and uniform over $q^k$, implying that every collection of $j$ variables between $1$ and $k$ is uniform over $q^j$. Thus, $H(S) = j$, which is integer valued.
\item $\mathbf{l, \: k \leq l \leq n} :$ Both the joint entropy $H(X)$ of cardinality $n$ and all subset entropies $H(X_A)$ of cardinality $k$ are equal to $k$. By monotonicity, every element of the power set with cardinality $l$ between $k$ and $n$ must also have $H(S) = k$, which is integer valued.
\end{itemize}
\end{proof}

By Lemma 2, the upper bound on subset entropies, and the fact that entropy is a submodular set function, entropy satisfies the requirements of a rank function for any distribution meeting the constant bound with equality. By Theorem 1, we can construct the independent set of a matroid $\mathcal{C}_{k,n}$ by looking at all subsets $S$ of $X$ where $H(S) = |S|$. When the bound is met, the independent set consists of all subsets of $X$ with cardinality $\leq k$, making $\mathcal{C}_{k,n} \cong U_{k,n}$ (and consequently $\mathcal{C}_{k,n} \cong M(\mathcal{A}_{k,n})$ for the vector matroid of an MDS generator matrix). This takes care of the first claim, but says nothing about the form of the distributions meeting these criteria. To prove the second claim, we must explore the connection between independence in vector and probability spaces.\\

\begin{Theorem}[Mat\'{u}\v{s}]
If a matroid $M$ is $\mathbb{F}_q$-representable, then $M$ also describes the statistical independence relationships for a distribution with marginal support $q$ \cite{matus}.
\end{Theorem}
\begin{proof}
For a vector matroid $M(A)$ with rank function $f$, recall that $0 \leq f(I) \leq |I|, \quad I \subseteq E$, where $E$ is a collection of column labels for $A$. The dual space (collection of rows) for any $I$ columns must also have rank $f(I)$. The linear space spanned by these rows over $\mathbb{F}_q$ has $q^{f(I)}$ points. If we take each of these coordinate vectors and assign them a probability of $q^{-f(I)}$ uniformly, the entropy of this collection will be $H(X_I) = f(I) {\cdot} \log_q q = f(I), \: \forall I \subseteq [n]$. Since under these conditions the entropy is a rank function, we have a one-to-one correspondence between the rank of a vector space and a probability space. Since the ground sets are also the same, this implies $M(A)$ is isomorphic to the matroid describing the resulting distribution.
\end{proof}
\hfill\\
By Theorem 2, the maximizing distributions for \textit{Cohesion-k}, whose matroid structure is the same as $M(\mathcal{A}_{k,n})$, can be generated by $\mathcal{A}_{k,n}$ (as evidenced by Table \ref{tbl3} and the Reed-Solomon code from the Section \ref{sec3:b} example). This result is almost enough to conclude the form of Cohesion maximizers in full generality, but one final piece is needed. Construction via Reed-Solomon code generalizes to any $n$ and $k$, as long as $n = p^m$, but what about $n \neq p^m$? \\

\begin{Theorem}
There exists an integer $q$ such that a distribution with marginal support $q$ achieves the upper bound on \textit{Cohesion-k} for $n$ variables.
\end{Theorem}
\begin{proof}
Since the matroid structures are the same, it suffices to show that $U_{k,n}$ is $\mathbb{F}_q$-representable for some integer $q$. A \textit{\textbf{family}} of subsets of $X$ is a finite sequence $(A_1,A_2,...,A_m)$ such that each member $A_j, \: j \in \{1,2,...,m\}$, $A_j \subseteq X$ need not be distinct. A \textit{\textbf{transversal}} is a subset $\{e_1,e_2,...,e_m \}$ of $X$ such that each $e_j \in A_j, \: \forall j \in J$. If $S \subseteq X$, then $S$ is a \textit{\textbf{partial transversal}} if for some subset $K$ of $J$, $S$ is a transversal of $(A_j : j \in K)$. For a family $\mathcal{A}$ of subsets, let $\mathcal{I}$ be the set of all partial transversals of $\mathcal{A}$. Then $\mathcal{I}$ defines the independent set of a 
\textbf{\textit{transversal matroid}} on $X$. The class of uniform matroids is a subclass of transversal matroids (ex: family $\mathcal{A} = (A_1,A_2,...,A_k)$ such that each $A_i = [n]$ defines a transversal matroid that is isomorphic to $U_{k,n}$). By Corollary $12.2.17$ in \cite{oxley}, for all transversal matroids $M$, there exists an integer $n(M)$ such that $M$ is representable over every extension field of $\mathbb{F}$ having at least $n(M)$ elements.
\end{proof}
\hfill

This concludes the proof in full generality, and solidifies the connection between matroids, MDS codes, and maximizers of Cohesion. Having accomplished what we set out to do by establishing this connection, what's left is to show from a broader perspective where these findings may be applicable.

\section{RELATED WORK}
\label{sec5}

\begin{figure*}
  \includegraphics[width=\textwidth]{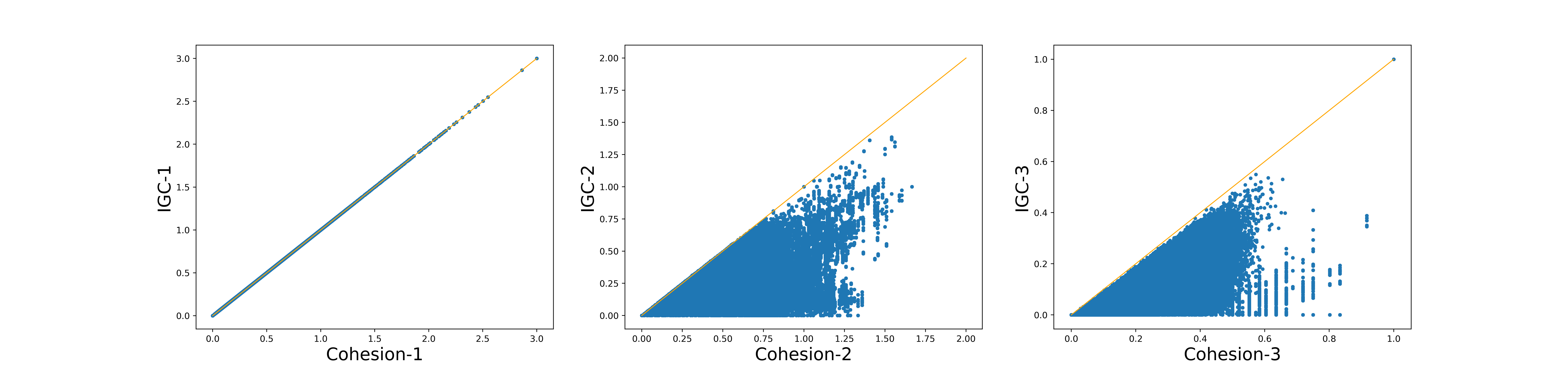}
  \caption{Plots of normalized \textit{Cohesion-k} vs $D(p||p^{(k)})$ (labeled IGC-k along the y-axis) for four binary variables. The orange line represents the upper bound of Equation \ref{eq4}. Each point appearing below this line represents lower order dependence that was accounted for by Cohesion, but not the divergence.}
  \label{fig2}
\end{figure*}

With the intention of quantifying emergence in complex systems, or how the whole deviates from the sum of its parts, Ay et al. \cite{nihat_geom} developed an approach based on the principles of information geometry. In this framework, the parts of a system are represented by the parameters of a hierarchical model $\mathcal{E}_k$. For a specified value of $k$, each parameter corresponds to a $k$th order marginal of $p$. The deviation from this hierarchical model is given by the KL-Divergence $D_{KL}(p||\mathcal{E}_{k})$. Due to the closure property of hierarchical models, this deviation is equivalent to $D_{KL}(p||p^{(k)})$, where $p^{(k)}$ is the maximum entropy projection for $\mathcal{E}_k$. This distribution is restricted to have the same $k$th order marginals as $p$. Each of these divergence measures $D_{KL}(p||p^{(k)})$ was shown to be upper-bounded by a normalized version of \textit{Cohesion-k}

\begin{equation}\label{eq4}
D_{KL}(p||p^{(k)}) \leq \frac{1}{{n-1 \choose k-1}}\mathcal{C}^{(k)}
\end{equation} 

This bound is loose in general, as evidenced by results on a tighter upper bound:

\begin{equation*}
D(p||p^{(k)}) \leq H_p(k) + (N-k)h_p(k) - H_p(N)
\end{equation*}

For a better understanding of the notation used above (what $H_p(\cdot)$ and $h_p(\cdot)$ represent), see section $3$ of \cite{nihat_geom}. Other than these bounds, no results on the global maximizers of $D(p||p^{(k)})$ were given for general $k$ until a recent paper by Matus \cite{matus_div}. An upper bound on the divergence to hierarchical log-linear models was introduced for any family of subsets $\mathcal{A}$ over $X$, with an analysis of its tightness and achievability. When $\mathcal{A}$ is the family of all cardinality $k$ subsets, this divergence is the same as $D(p||p^{(k)})$. This special case was studied as an example of achievability, and a subset of the global maximizers were arrived at by matroid-based arguments. These distributions, referred to as partition representable distributions by Matus \cite{prep}, are exactly the global maximizers of \textit{Cohesion-k}. These findings imply that the upper bound of Equation \ref{eq4} is tight for the extrema of both measures over a sufficiently large field, but what about more generally? By the nature of $D(p||p^{(k)})$ as a KL-Divergence, any distribution that can be represented solely by its $k$th order marginals receives zero weight, meaning no lower-order information $(\leq k)$ contributes to the measure. For example, the maximizers of $D(p||p^{(1)})$ (the maximally redundant distributions of Table \ref{tbl1:sub1}) would be in the set of global minimizers for $D(p||p^{(2)})$. This is not the case for \textit{Cohesion-k}, which only assigns zero weight when all variables are independent. As a result, Cohesion has preference for linear dependence structures, since $(k-1)$th order interactions receive almost as much weight as $k$th order. The gap in the bound between each measure can thus be accounted for by their different treatment of lower order information. To illustrate these differences further, we compare \textit{Cohesion-k} and $D(p||p^{(k)})$ empirically for four binary variables.\\

Directing our attention to the second plot in Figure \ref{fig2}, recall that the maximizing distributions for \textit{Cohesion-2} over four binary variables was given in Table \ref{tbl2}. By the result of Matus, we know that over a sufficiently large field, $D(p||p^{(2)})$ and \textit{Cohesion-2} have the same global maximizers, but what about when the support is restricted? When optimizing over a discretized simplex of bin size $1/12$, the maximizing distribution for $D(p||p^{(2)})$ over four binary variables is

\begin{table}[!htb]
    	\centering
        \caption{Binary (local) Maximizers for $D(p||p^{(2)})$}
		\begin{tabular}{cccc|c|}
        	\boldmath$X_0$ & \boldmath$X_1$ & \boldmath$X_2$ & \boldmath$X_3$&\boldmath$Pr$\\
			\hline
            0 & 0 & 0 & 0 & 1/4\\
            0 & 1 & 1 & 1 & 1/4\\
            1 & 0 & 1 & 1 & 1/6\\
            1 & 1 & 0 & 1 & 1/6\\
            1 & 1 & 1 & 0 & 1/6\\
		\end{tabular}
        \label{tbl4}
\end{table}

While this might not be the global maximizer, we can still gain insight from comparing it to maximizers of \textit{Cohesion-2}.  This distribution does not have a linear dependence structure, since no marginal variables have entropy of $1$ bit. It does, however, have more third order information (1.44 bits) than the normalized redundant-synergy maximizer of \textit{Cohesion-2} (1 bit). When the support is large, it is possible to have both linear dependence and maximal information of a particular order (Table \ref{tbl3}), which is the main reason why these measures converge at the extremes.

\section{DISCUSSION AND FUTURE WORK}
\label{sec6}
In addition to its appearance in this work, coding theory has been referenced in a few recent papers on MMI and information decomposition, hinting that a meaningful connection may exist between the two fields. Secret sharing, as alluded to in Section \ref{sec3:b}, is a sub-discipline of coding theory and cryptography concerned with the secure transmission of information. Each message is transformed into a codeword such that no information about the message or its encoding process can be recovered by an eavesdropper via repeated observation (and under certain computational assumptions). Ideas from secret sharing have been used to argue for an expanded definition of local positivity in the partial information decomposition \cite{SSS}, and to define new measures of MMI-based capacity bounds in secret key agreement \cite{SKA}. MDS codes, under the secret sharing interpretation, correspond to perfect/ ideal secret sharing schemes where the share is the same size as the secret and no collection of $(k-1)$ users can recover it (assuming the secret is size $k$). Although this is the first time MDS codes and perfect/ideal secret sharing schemes have been affiliated with Cohesion measures, it is not the first time they have appeared in the study of MMI.\\

Since many potential real world applications of these measures will have data with fixed support much smaller than $q=n$, it will be interesting to further study the maximizers of Cohesion over small/ restricted fields. From the limited experimental evidence we gathered in comparing to $D(p||p^{(k)})$ (and some additional simplex search), it seems that Cohesion maximizers are always related to linear dependence structures over $\mathbb{F}_q$, with the matroid structure as close to uniform as allowable over the fixed field. Such structures seem conceptually close to near-MDS (NMDS) and almost-MDS (AMDS) codes over small fields \cite{nmds}\cite{amds}.


\addtolength{\textheight}{-4cm}   



\section*{APPENDIX}

\subsection*{Proofs for additional polymatroid bounds}
Each proof has similar construction to the constant upper bound of Section \ref{sec4}. Each term in the sum is upper-bounded by its cardinality, and the joint entropy is lower bounded through monotonicity.

\begin{equation*}
\begin{split}
\mathcal{C}^{(1)} + \mathcal{C}^{(3)} &= \sum\limits_{i=1}^n H(X_i) +  \sum\limits_{i=1}^n H(X_{E/i}) - 4 {\cdot} H(X)\\
&\leq n(1) + n(n-1) - 4(n-1)\\
&\leq 4
\end{split}
\end{equation*}

\begin{equation*}
\begin{split}
\mathcal{C}^{(2)} + 3{\cdot}\mathcal{C}^{(1)} &= \sum\limits_{A=1}^{n \choose 2} H(X_A) +  \sum\limits_{i=1}^n 3{\cdot}H(X_i) - 6 {\cdot} H(X)\\
&\leq 2{n \choose 2 } + n(3) - 6(2)\\
&\leq 12
\end{split}
\end{equation*}

\begin{equation*}
\begin{split}
\mathcal{C}^{(2)} + 3{\cdot}\mathcal{C}^{(3)} &= \sum\limits_{A=1}^{n \choose 2} H(X_A) +  \sum\limits_{i=1}^n 3{\cdot}H(X_{E/i}) - 12 {\cdot} H(X)\\
&\leq 2{n \choose 2 } + n(n-1)(3) - 12(n-1)\\
&\leq 12
\end{split}
\end{equation*}

\bibliographystyle{IEEEtran}
\bibliography{test_bibdesk.bib}

\end{document}